\documentclass[10pt,a4paper]{article}
\usepackage[utf8]{inputenc}
\usepackage[T1]{fontenc} 
\usepackage{lmodern}

\usepackage[english]{babel}

\usepackage{multicol}

\usepackage{graphics, graphicx}

\usepackage{amsmath}
\usepackage{amsfonts}
\usepackage{amssymb}
\usepackage{amsthm}

\def\scpmsat{{\sc{Planar Monotone $(2,\,3)$-Sat-3}}}
\def\bpmsat{{\sc{Restricted Planar Monotone $(2,\,3)$-Sat-4}}}
\def\restrsat{{\sc{Restricted Planar Monotone $(2,\,3)$-Sat}}}
\def\plan3satE5{{\sc{Planar Monotone 3-Sat-E5}}}
\def\plantc3sat4{{\sc{4-Bounded Planar 3-Connected 3-Sat}}} 
\def\plansat{{\sc{Planar 3-Sat}}}
\def\sat34{{\sc{3-Sat-4}}}
\def\multisat{{\sc{Sat}$^*$}}
\def\planmultisat{{\sc{Planar Monotone 3-\multisat}}}

\newtheorem{theorem}{Theorem}

\newtheorem{corollary}{Corollary}
\newtheorem{definition}{Definition}

\begin{document}
	\title{On planar variants of the monotone satisfiability problem with bounded variable appearances}
	\author{Andreas Darmann, Janosch Döcker, Britta Dorn}
	\date{\today}
	\maketitle

\begin{abstract}
We show $\mathcal{NP}$-completeness for several planar variants of the monotone satisfiability problem with bounded variable appearances. With one exception the presented variants have an associated bipartite graph where the vertex degree is bounded by at most four. Hence, a planar and orthogonal drawing for these graphs can be computed efficiently, which may turn out to be useful in reductions using these variants as a starting point for proving some decision problem to be $\mathcal{NP}$-hard. 
\end{abstract}

\section{Introduction}
The satisfiability problem for boolean formulae in conjunctive normal form (\textsc{Sat}), where each clause is monotone, i.\,e., all literals in a clause are positive or all of them are negative, is known to be $\mathcal{NP}$-complete \cite{Gold1978}, even if every clause contains exactly three \emph{distinct} literals (see the work by Li \cite{Li1997}). Further, Tovey \cite{Tovey1984} proved \sat34{}, where each clause contains exactly three distinct variables and every variable appears in at most four clauses, to be $\mathcal{NP}$-complete and showed that instances of this problem where every variable appears at most three times are always satisfiable. Darmann and Döcker~\cite{Darmann2016} show that \textsc{Monotone} \sat34, i.e., the restriction of \sat34{} to instances in which each clause is monotone remains $\mathcal{NP}$-complete. For a general information on computational complexity and, in particular, the concept of $\mathcal{NP}$-completeness we refer to Garey and Johnson \cite{Garey1979}.

Unless otherwise stated, we write clauses as subsets of a finite set $\mathcal{V}$ of variables; only in one of the problems considered (see Definition~\ref{def:duplicates}) we use multisets in order to allow for duplicates of literals within a clause. A $k$-clause contains exactly $k$ distinct variables and a clause is called monotone if either all contained literals are positive or all of them are negative, respectively. A mixed clause is a clause which is not monotone, i.e., it contains at least one positive and at least one negative literal. A clause is called positive (negative) if it contains only positive (negative) literals. We will refer to the replacement rule Gold~\cite[p.\,314f]{Gold1978} described to transform an instance $\mathcal{I}$ of the satisfiability problem for boolean formulae in conjunctive normal form into an instance $\mathcal{I}'$, such that each clause in $\mathcal{I}'$ is monotone and the two instances $\mathcal{I}$ and $\mathcal{I}'$ are equisatisfiable, as \emph{Gold's rule}. In set notation this rule looks as follows: A mixed clause $C = C^+ \cup C^-$ ($C^+$ contains the positive literals of $C$ and $C^-$ the negative literals, respectively) is replaced by the two monotone clauses $C^+ \cup \{u\}$ and $C^- \cup \{\bar{u}\}$, where $u$ is a new variable. 

In this paper, we consider planar variants of the satisfiability problem, i.\,e, instances $\mathcal{I} = (\mathcal{V}, \mathcal{C})$, where $\mathcal{V}$ is a set of variables and $\mathcal{C}$ a collection of clauses $C_j$ so that the graph~$G_{\mathcal{V},\mathcal{C}} := (V,\,E)$ with $V := \mathcal{V} \cup \mathcal{C}$ and 
\[
E := \{\{v_i,\,C_j\} : v_i \in C_j \vee \bar{v}_i \in C_j\}
\] 
is planar. Lichtenstein \cite{Lichtenstein1982} proved the planar variant of the satisfiability problem to be $\mathcal{NP}$-complete, even for the case that each clause contains at most three variables and with the additional edges~$\{\{v_i,\,v_{i+1}\} : 1 \leq i < n\} \cup \{v_n,\, v_1\}$ where $n := |V|$ -- this variant is called \plansat. The restriction of this problem -- without the additional edges -- to monotone instances remains $\mathcal{NP}$-complete, even if a rectilinear representation is given \cite{Berg2010}. Other related $\mathcal{NP}$-complete problems -- without the monotonicity requirement -- are, e.\,g., \plantc3sat4\ \cite{Kratochvil1994} and a variant of Dahlhaus et al. \cite{Dahlhaus1994} where every variable appears exactly three times (more details on these variants are given in the following section when they are used in reductions). We present hardness proofs for planar and monotone variants of the satisfiability problem with bounds on the number of variable appearances.  \newline
In what follows, by \textsc{Sat}-$s$ (respectively \textsc{Sat}-E$s$) we denote the restriction of \textsc{Sat} to instances in which  each variable appears at most $s$ times (respectively exactly $s$ times), $s \in \mathbb{N}$. 

\section{Bounded planar variants}

In this section we consider planar variants of the monotone satisfiability problem with the main focus on hardness with respect to bounds on variable appearances. 

\begin{definition}[\textsc{Planar Monotone $(2,\,3)$-Sat}]\
\label{def:PlanMon23Sat}
\begin{itemize}
\item \textbf{Input:} A set of variables $\mathcal{V} = \{v_1,\,v_2,\,\ldots,\,v_n\}$, a collection of clauses $\mathcal{C} = \{C_1,\,C_2,\,\ldots,\,C_m\}$ and a graph $G_{\mathcal{V},\mathcal{C}} := (V,\,E)$ with $V := \mathcal{V} \cup \mathcal{C}$ and $E := \{\{v_i,\,C_j\} : v_i \in C_j \vee \bar{v}_i \in C_j\}$, so that the following properties hold:
\begin{enumerate}
\item The bipartite graph $G_{\mathcal{V},\mathcal{C}}$ is planar.
\item Each clause contains two or three distinct literals, either all or none of them are negated.
\end{enumerate}
\item \textbf{Question:} Is the collection of clauses $\mathcal{C}$ satisfiable?
\end{itemize}
\end{definition}

\begin{theorem}\label{the:scpmsat}
\scpmsat\ is $\mathcal{NP}$-complete.
\end{theorem}

\begin{proof}
The problem is in $\mathcal{NP}$, since it is a special case of the satisfiability problem for boolean formulas in conjunctive normal form. We show that \scpmsat\ is $\mathcal{NP}$-hard by a reduction from a restricted version of Lichtenstein's \plansat\ \cite{Lichtenstein1982}. Dahlhaus et al. \cite{Dahlhaus1994} have shown that the latter problem remains $\mathcal{NP}$-hard if each clause contains two or three literals, each variable appears in exactly three clauses, with one of its literals appearing in two clauses and the other literal in one clause. It is easy to see that applying Gold's rule to each mixed clause of such an instance preserves planarity and respects the bounds on variable appearances. Therefore, \scpmsat\ is $\mathcal{NP}$-hard. 
\end{proof}

Note that for any instance $\mathcal{I}=(\mathcal{V},\mathcal{C})$ of \textsc{Planar Monotone $(2,\,3)$-Sat-$3$} and any variable $x \in \mathcal{V}$ which appears at most twice, we can construct an equisatisfiable instance $\mathcal{I}'=(\mathcal{V}',\mathcal{C}')$ of \textsc{Planar Monotone $(2,\,3)$-Sat-$3$} in which the number of appearances of $x$ is increased by exactly one by adding to $\mathcal{C}$ the clauses $\{x,u,v\}$, $\{u,v\}$, $\{\bar{u},\bar{v}\}$, where $u,v$ are newly introduced variables. As a consequence, we get the following corollary.

\begin{corollary}\label{cor:plan23SatE3}
 \textsc{Planar Monotone $(2,\,3)$-Sat-E$3$} is $\mathcal{NP}$-complete.
\end{corollary}

\begin{definition}[\restrsat]\
\label{def:restrSAT}
\begin{itemize}
\item \textbf{Input:} A set of variables $\mathcal{V} = \{v_1,\,v_2,\,\ldots,\,v_n\}$, a collection of clauses $\mathcal{C} = \{C_1,\,C_2,\,\ldots,\,C_m\}$ and a graph $G_{\mathcal{V},\mathcal{C}} := (V,\,E)$ with $V := \mathcal{V} \cup \mathcal{C}$ and $E := \{\{v_i,\,C_j\} : v_i \in C_j \vee \bar{v}_i \in C_j\}$, so that the following properties hold:
\begin{enumerate}
\item The bipartite graph $G_{\mathcal{V},\mathcal{C}}$ is planar.
\item Each clause contains two or three distinct literals, either all or none of them are negated. Every 3-clause contains only positive literals.
\item  Each variable appears negated exactly once.
\end{enumerate}
\item \textbf{Question:} Is the collection of clauses $\mathcal{C}$ satisfiable?
\end{itemize}
\end{definition}

\begin{theorem}\label{the:bpmsat}
\bpmsat\ is $\mathcal{NP}$-complete, even if every variable appears at least three times.
\end{theorem}

\begin{proof}
As in Theorem \ref{the:scpmsat} the problem is in $\mathcal{NP}$ and we use the variant of Dahlhaus et al. \cite{Dahlhaus1994} for the  reduction: each clause contains two or three literals, each variable appears in exactly three clauses, with one of its literals appearing in two clauses and the other literal in one clause. We can compute a planar and orthogonal drawing of the graph on a grid of size $n \times n$ in linear time using the algorithm of Biedl and Kant~\cite{Biedl1998}. On this drawing we apply the local replacements shown in Figure \ref{fig:orientations}, 
\begin{figure}
\centering
\includegraphics[width=\textwidth]{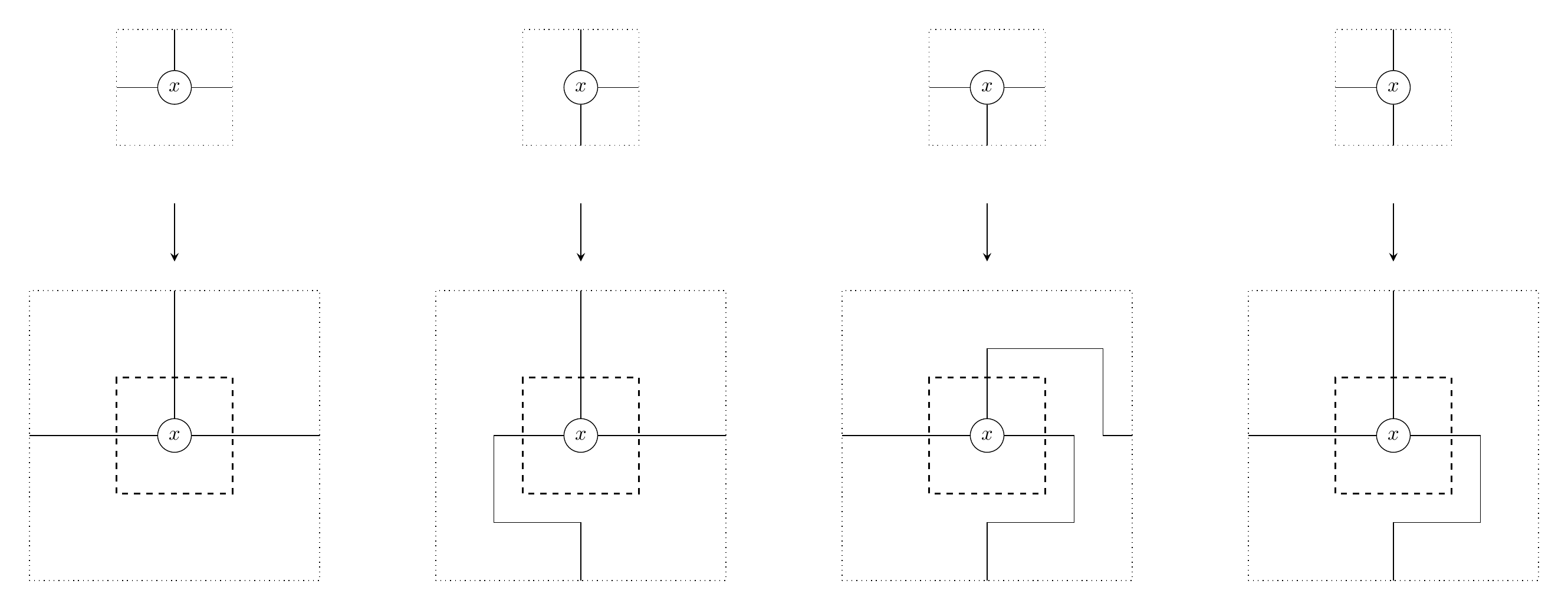}
\caption{Local replacement of a variable vertex $x$. After application of this rule the drawing of the vertex looks locally as shown in the dashed square.}
\label{fig:orientations}
\end{figure} 
so that for each variable the outgoing edges of the corresponding vertex are locally drawn identically (see the dashed squares in Figure \ref{fig:orientations}). 
\begin{figure}
\centering
\includegraphics[width=\textwidth]{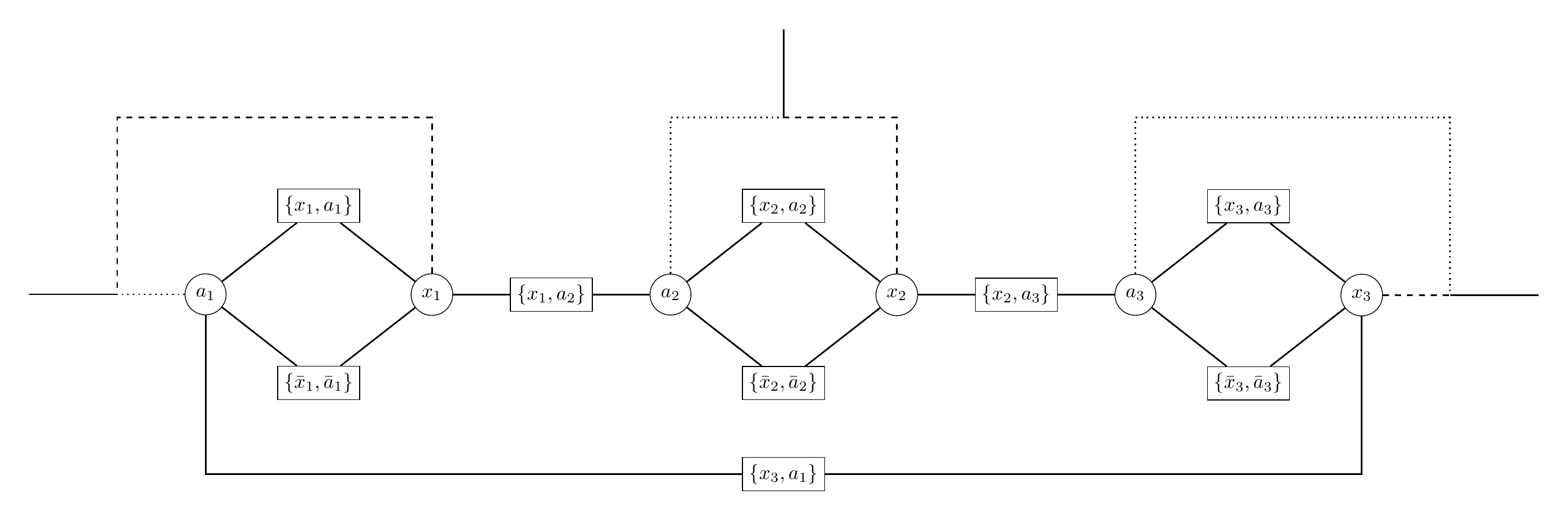}
\caption{Creating an equisatisfiable planar monotone instance with bounded variable appearances. The dotted line is used if the corresponding appearance of the variable is negated; otherwise the dashed line is used.}
\label{fig:gadget}
\end{figure} 
For every variable we replace the dashed square with the construction shown in Figure \ref{fig:gadget} (this gadget is inspired by the one given by Dahlhaus et al. \cite[p.\,18]{Dahlhaus1994} and the one by de Berg and Khosravi \cite[p.\,6]{Berg2010}). For every appearance of a variable $x$ we create two new variables $x_i$ and $a_i$ and replace $x$ with $x_i$. In order to satisfy the two clauses between $a_i$ and $x_i$ in the gadget (see Figure \ref{fig:gadget}) we have to assign opposite truth values to $a_i$ and $x_i$. The ring structure forces us to assign the same truth value to all $x_i$, and consequently the opposite truth value to all $a_i$. If $x_i$ appears non-negated in the corresponding clause $C_{x_i}$ of the original instance, we just use the outgoing dashed line of $x_i$. Otherwise, we replace $\bar{x}_i$ with $a_i$ in $C_{x_i}$, and replace the edge $\{x_i,\, C_{x_i}\}$ with $\{a_i,\, C_{x_i}\}$. The graph remains planar since we only need to reroute the dashed line to $a_i$ by using the corresponding dotted line instead. Now, it is not hard to see that the resulting instance is equisatisfiable. Every clause contains two or three distinct literals since we introduced only monotone 2-clauses. All clauses of the original instance are replaced with positive clauses containing the same number of literals. Hence, every clause is monotone and all 3-clauses are positive. We replace every variable with the construction described above, so every variable appears either three or four times; either $a_i$ appears three times and $x_i$ four times or the other way round. Moreover, recalling that only the clauses introduced in the gadget contain negative literals in the final instance, we can conclude that every variable appears exactly once negated.  
\end{proof}

Now, for any instance $\mathcal{I}=(\mathcal{V},\mathcal{C})$ of \bpmsat\   we can construct an equisatisfiable instance  $\mathcal{I}'=(\mathcal{V}',\mathcal{C}')$ of \restrsat-\textsc{E$4$} as follows:
for any variable $x \in \mathcal{V}$ which appears three times add to $\mathcal{C}$ the clauses $\{x,a,b\}$, $\{a,c,d\}$,
$\{b,c,d\}$, $\{a,b\}$, $\{\bar{a},\bar{b}\}$, $\{c,d\}$, $\{\bar{c},\bar{d}\}$, where $a, b, c, d$ are newly introduced variables. As a consequence, from Theorem~\ref{the:bpmsat} we get the following result.

\begin{corollary}
 \restrsat-\textsc{E$4$} is $\mathcal{NP}$-complete.
\end{corollary}

\bigskip

An interesting question is whether or not we can replace a monotone 2-clause in an instance of \textsc{Planar Monotone $(2,\,3)$-Sat} with monotone 3-clauses such that the result is an equisatisfiable instance of this problem, i.\,e., the corresponding graph remains planar. Unfortunately, the replacement rules presented in previous work \cite{Darmann2016} do not preserve planarity. In Figure \ref{fig:bipartite} this is shown for the rules $\mathcal{R}_i$ in Theorem 1 of this article: Rule $\mathcal{R}_1$ replaces a clause $\{x,\,y\}$ with  clauses
\[
\{x,\,y,\,u\},\,\{x,\,y,\,v\},\,\{x,\,y,\,w\},\,\{\bar{u},\,\bar{v},\,\bar{w}\},
\]
where $u,\,v,\,w$ are new variables; a clause $\{\bar{x},\,\bar{y}\}$ is handled analogously (this rule is due to Li \cite[p.\,295]{Li1997}). Since the other rules use $\mathcal{R}_1$ as a ``subroutine''(in the case of $\mathcal{R}_3$ indirectly by building on $\mathcal{R}_2$), none of the rules preserves planarity. An application of the replacement rule presented in the proof of Theorem 2 in the article \cite[p.\,4]{Darmann2016} mentioned above results in a graph with $K_{3,3}$ as a minor as well, since the clauses 3, 4, 5 and 6 together with the variables contained in these clauses induce a subgraph isomorphic to the graph shown on the left in Figure \ref{fig:bipartite}.

\begin{figure}
\centering
\includegraphics[width=\textwidth]{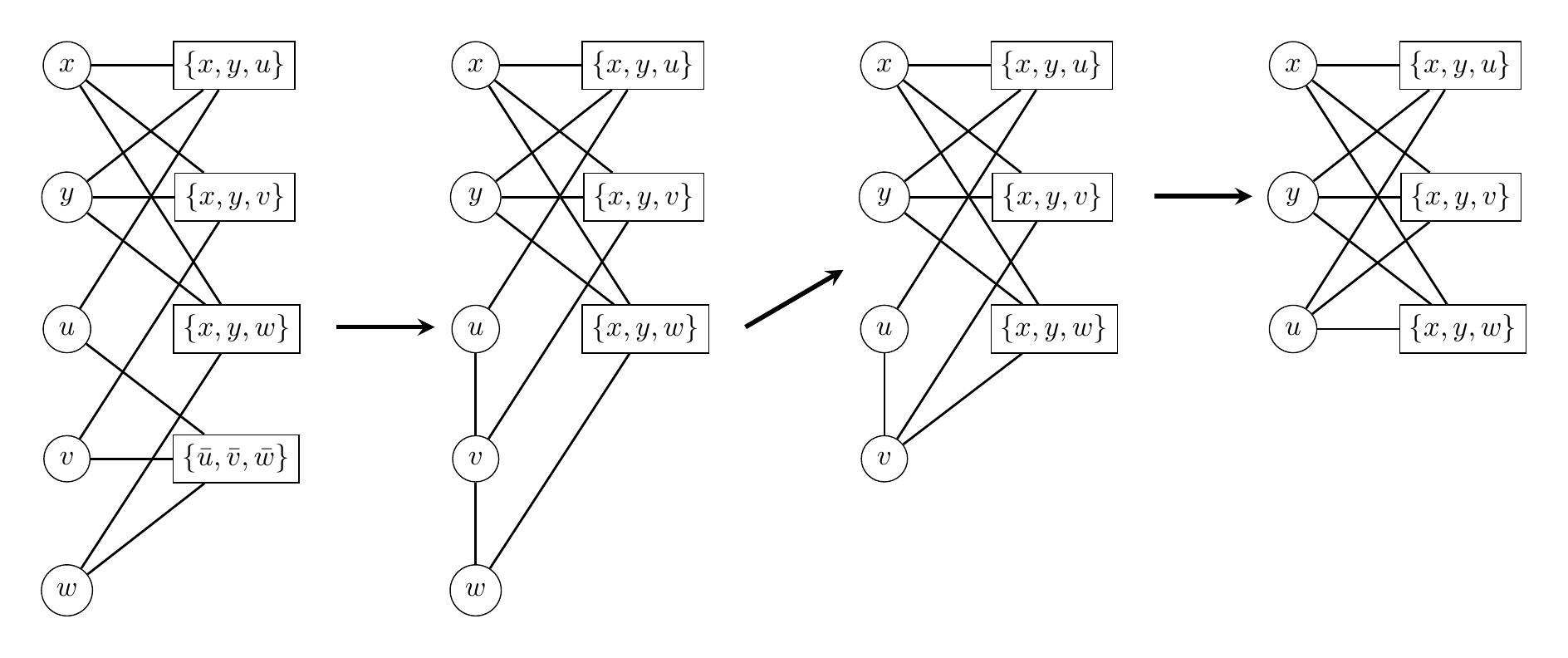}
\caption{The replacement rules $\mathcal{R}_i$ do not preserve planarity, since already the bipartite graph associated with $\mathcal{R}_1(\{x,y\})$ has $K_{3,3}$ as a minor.}
\label{fig:bipartite}
\end{figure} 

Of course, we can achieve the goal described above by allowing duplicates of a variable in a clause; the clauses are now multisets for this reason. We denote this variation of \textsc{Sat} in which clauses are multisets of variables by \multisat (again, \multisat-E$s$ indicates that each variable appears exactly $s$ times). With this relaxation we show that the problem is $\mathcal{NP}$-complete already for the case that each variable appears exactly $4$ times. 
Formally, we define the following problem.

\begin{definition}[\planmultisat]\
\label{def:duplicates}
\begin{itemize}
\item \textbf{Input:} A set of variables $\mathcal{V} = \{v_1,\,v_2,\,\ldots,\,v_n\}$, a collection of clauses $\mathcal{C} = \{C_1,\,C_2,\,\ldots,\,C_m\}$ and a graph $G_{\mathcal{V},\mathcal{C}} := (V,\,E)$ with $V := \mathcal{V} \cup \mathcal{C}$ and $E := \{\{v_i,\,C_j\} : v_i \in C_j \vee \bar{v}_i \in C_j\}$, so that the following properties hold:
\begin{enumerate}
\item The bipartite graph $G_{\mathcal{V},\mathcal{C}}$ is planar.
\item Each clause contains exactly three literals, either all or none of them are negated (duplicates are allowed).
\end{enumerate}
\item \textbf{Question:} Is the collection of clauses $\mathcal{C}$ satisfiable?
\end{itemize}
\end{definition}

\begin{theorem}\label{the:planmultisatE4} \planmultisat-\textsc{E}$4$ is $\mathcal{NP}$-complete. \end{theorem} 

\begin{proof}Reduction from  \textsc{Planar Monotone $(2,\,3)$-Sat-E$3$} (see Corollary~\ref{cor:plan23SatE3}). Given an instance $\mathcal{I}=(\mathcal{V},\mathcal{C})$
of \textsc{Planar Monotone $(2,\,3)$-Sat-E$3$}, construct instance $\mathcal{I}'=(\mathcal{V}',\mathcal{C}')$
of  \planmultisat-\textsc{E}$4$ as follows. Replace each $2$-clause $\{x,y\}$ of
$\mathcal{C}$ by the two $3$-clauses $\{x,y,z\}$, $\{\bar{z},\bar{z},\bar{z}\}$
($z$ is a new variable). Recall that each variable $x \in \mathcal{V}$ appears exactly three times in 
collection $\mathcal{C}$. For each variable  $x \in \mathcal{V}$
 add the clauses $\{x,u,u\}$, $\{u,u,v\}$, $\{v,v,v\}$,
where $u,v$ are new variables. \newline
Obviously, $\mathcal{I}$ is a ``yes''-instance of \textsc{Planar Monotone $(2,\,3)$-Sat-E$3$} if
and only if $\mathcal{I}'$ is a ``yes''-instance of  \planmultisat-\textsc{E}$4$.
It is also easy to see that the graph $G_{\mathcal{V}',\mathcal{C}'}$
is planar. \end{proof}

\begin{theorem}\label{the:planmultisatE5}
\planmultisat-\textsc{E}$5$ is $\mathcal{NP}$-complete, even if restricted to instances in which the graph $G_{\mathcal{V},\mathcal{C}}$ is biconnected. 
\end{theorem}

\begin{proof}
Again, the problem is clearly in $\mathcal{NP}$. This time we use a reduction from \plantc3sat4\ \cite{Kratochvil1994} in order to show $\mathcal{NP}$-hardness. In an instance $\mathcal{I} := (\mathcal{V},\, \mathcal{C})$ of this variant,   every variable appears at least three and at most four times and every clause contains exactly three distinct variables. Further, the bipartite graph $G_{\mathcal{V},\mathcal{C}} := (V,\,E)$ with $V := \mathcal{V} \cup \mathcal{C}$ and $E := \{\{v_i,\,C_j\} : v_i \in C_j \vee \bar{v}_i \in C_j\}$ is planar and vertex 3-connected. As in Theorem \ref{the:bpmsat}, we may assume an orthogonal drawing to be given -- otherwise we can compute it in linear time -- and replace every variable vertex locally. In Figure \ref{fig:gadget} the construction for a vertex of degree~3 is given. The adaptation of the gadget for a vertex of degree 4 is straightforward (see Figure \ref{fig:gadget2}). For details on the gadgets see the proof of Theorem \ref{the:bpmsat}. Every clause of the original instance $\mathcal{I}$ contains exactly three variables, and is monotone after the local replacements. Further, every variable vertex of $\mathcal{I}$ is replaced. Thus, by showing that after using a gadget $\mathcal{G}$ for a local replacement, we can duplicate literals in the clauses of $\mathcal{G}$ so that every clause of $\mathcal{G}$ contains exactly three literals and every variable of $\mathcal{G}$ appears exactly 5 times in the final instance, these properties follow globally for the final instance. Recall that duplication of a literal within a clause is now possible. Now, the reasoning is identical for both gadgets: We consider the ring structure of a gadget in the clockwise direction. For $a_i$ and $x_i$ recall that depending on the instance either $a_i$ or $x_i$ has degree~3 and the other one has degree 4. We always duplicate $x_i$ in the clause on the right of the corresponding variable vertex (on the bottom of Figure \ref{fig:gadget2} this clause is drawn on the left). For the two clauses drawn between $a_i$ and $x_i$ we have two cases: If the degree of the vertex $x_i$ is 5 after the just described duplication, we duplicate $a_i$ in the upper clause and $\bar{a}_i$ in the lower clause. Since in this case $a_i$ must have had degree 3 before the duplication, it now has degree 5 as well. Otherwise both $x_i$ and $a_i$ have degree 4 and we duplicate, e.\,g., $x_i$ in the upper clause and $\bar{a}_i$ in the lower clause. Again both variable vertices have now degree~5. Doing this for every pair $(a_i,\,x_i)$ in the ring structure yields an equisatisfiable construction with the desired properties. 

The resulting graph is biconnected since the graph in the instance $\mathcal{I}$ is vertex 3-connected and the gadget we used for the local replacement in the construction of the instance of \planmultisat-\textsc{E}$5$ is (obviously) biconnected.
\begin{figure}
\centering
\includegraphics[width=\textwidth]{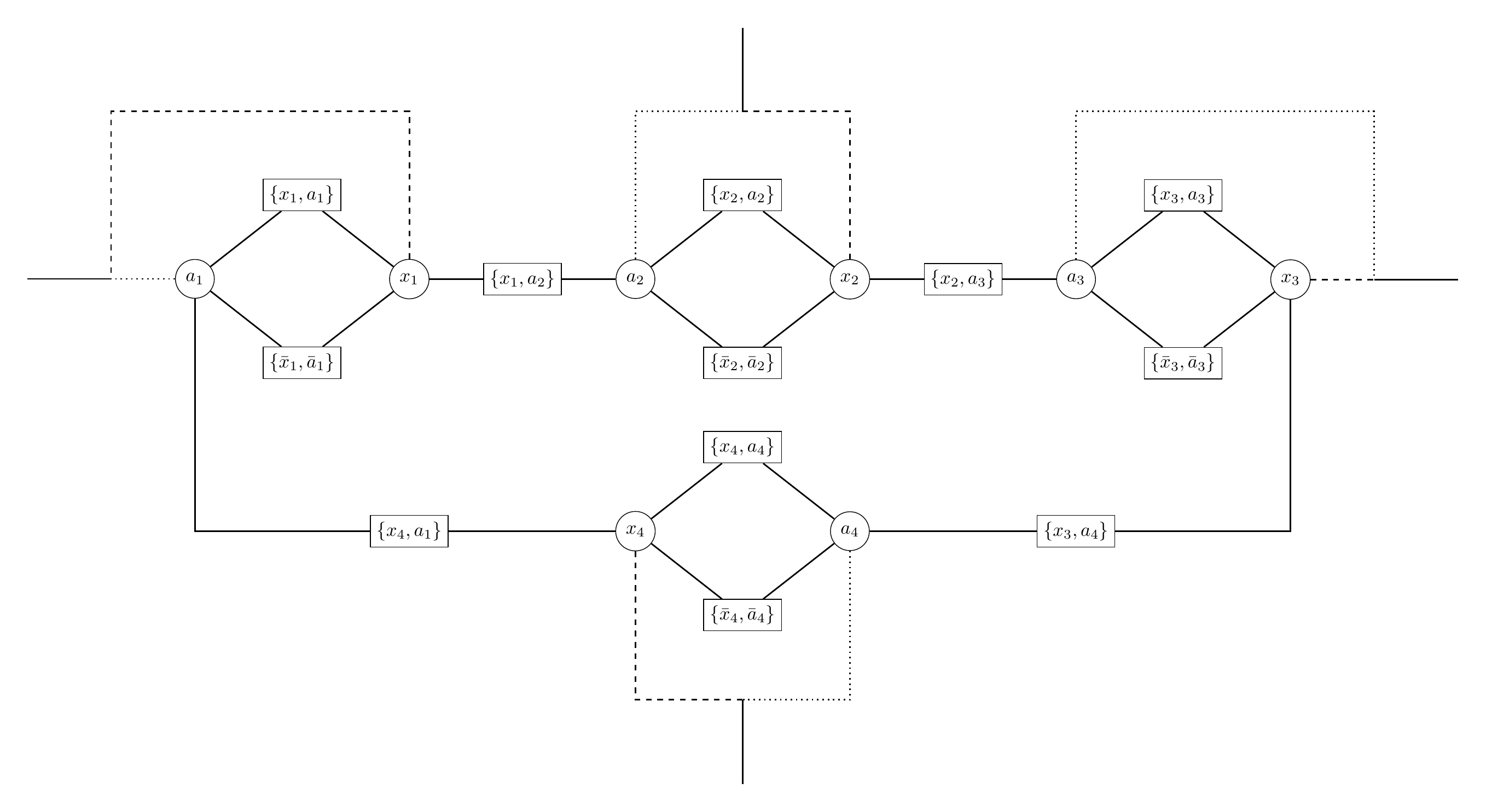}
\caption{Gadget for vertices of degree 4.}
\label{fig:gadget2}
\end{figure}
\end{proof}

\section{Conclusion}

We have shown $\mathcal{NP}$-completeness for several planar variants of the monotone satisfiability problem for boolean formulae in conjunctive normal form involving bounds on the number of times a variable may appear in the formula. The variants considered in Theorem \ref{the:scpmsat} and \ref{the:bpmsat} have an associated bipartite graph with vertex degree bounded by three and four, respectively. Further, we have shown that these variants remain hard, even if restricted to instances in which every variable appears exactly as often in the collection of clauses as the upper bound for the variables allows. Thus, planar and orthogonal drawings of these graphs exist and can be computed efficiently (see the work by Biedl and Kant~\cite{Biedl1998}), which may be useful when using these variants as a starting point for a reduction in a $\mathcal{NP}$-hardness proof. Moreover, the variant proven to be $\mathcal{NP}$-complete in Theorem \ref{the:planmultisatE4} has the property that each clause contains exactly three, not necessarily distinct, literals and every variable appears exactly four times in the collection of clauses. Finally, the variant shown to be $\mathcal{NP}$-complete in Theorem~\ref{the:planmultisatE5} differs from the one considered in Theorem \ref{the:planmultisatE4} in the way that every variable is required to appear exactly five times and the corresponding planar graph is biconnected. 

An interesting open question remains: If we require every clause to contain exactly three \emph{distinct} variables, is there a number $s \in \mathbb{N}$ such that {\sc Planar Monotone 3-Sat-$s$} is $\mathcal{NP}$-hard and if so, what is the smallest number with this property (clearly, $s \geq 4$)? In order to show the first part it would suffice to find a finite set of monotone 3-clauses similar to the sets used in the replacement rules in previous work \cite{Darmann2016} (these rules do not preserve planarity) and a planar drawing thereof such that the variable with the forced truth value is drawn on the outer face. The latter property ensures that we can add this variable to a 2-clause without destroying planarity.    

\bibliographystyle{alpha}
\bibliography{mylit}

\end{document}